\documentclass[11pt,reqno]{amsart}
\usepackage[T1]{fontenc}
\usepackage[utf8]{inputenc}
\usepackage{mathpazo}
\usepackage{amsmath,amssymb,amsthm,mathtools}
\usepackage[margin=1.05in]{geometry}
\usepackage{microtype}
\usepackage{xcolor}
\usepackage{graphicx,booktabs,flafter}
\usepackage[colorlinks=true,linkcolor=blue!45!black,citecolor=blue!45!black,urlcolor=blue!45!black]{hyperref}
\hypersetup{pdftitle={Boundary transmission and excursion geometry in the Lorentz mirror model: a numerical study},pdfauthor={}}
\numberwithin{equation}{section}
\newtheorem{theorem}{Theorem}[section]
\newtheorem{proposition}[theorem]{Proposition}
\newtheorem{lemma}[theorem]{Lemma}

\theoremstyle{definition}

\theoremstyle{remark}

\newcommand{\Z}{\mathbb Z}
\newcommand{\Pp}{\mathbb P_p}
\newcommand{\Ep}{\mathbb E_p}

\newcommand{\Cross}{\mathsf{Cross}}
\newcommand{\Col}{\mathsf{Col}}

\newcommand{\FirstAuthor}{Jian Gu}
\newcommand{\FirstAffiliation}{ESSEC}
\newcommand{\SecondAuthor}{Qin Hao}
\newcommand{\SecondAffiliation}{Polytechnic Institute of Paris}
\title[A numerical study of Lorentz mirror trajectories]{Boundary transmission and excursion geometry\\ in the Lorentz mirror model: a numerical study}
\author[\FirstAuthor]{\FirstAuthor}
\address{\FirstAffiliation}
\author[\SecondAuthor]{\SecondAuthor}
\address{\SecondAffiliation}
\makeatletter
\renewcommand{\@setauthors}{%
  \begingroup
  \trivlist
  \centering\footnotesize \@topsep30\p@\relax
  \advance\@topsep by -\baselineskip
  \item\relax
  \begin{minipage}[t]{.46\textwidth}\centering
    {\scshape\FirstAuthor\par}\smallskip
    {\normalfont\footnotesize\FirstAffiliation\par}
  \end{minipage}\hfill
  \begin{minipage}[t]{.46\textwidth}\centering
    {\scshape\SecondAuthor\par}\smallskip
    {\normalfont\footnotesize\SecondAffiliation\par}
  \end{minipage}
  \endtrivlist
  \endgroup
}
\renewcommand{\@setaddresses}{}
\makeatother

\newcommand{\WideNinety}{6149}
\newcommand{\WideNinetyLow}{5967}
\newcommand{\WideNinetyHigh}{6330}
\newcommand{\WideNinetyScaled}{12.01}
\newcommand{\WorstUpper}{0.0520}
\newcommand{\CylinderCensored}{0}
\begin{document}
\raggedbottom
\begin{abstract}
We study boundary transmission and excursion geometry in the planar Lorentz mirror model through finite-volume numerical experiments. Cylinder penetration quantiles grow approximately linearly with circumference over the simulated range, with a systematic upward drift in the normalized quantiles when vacancies are present. Estimated slab-transmission probabilities yield exponential confinement bounds with quantified statistical confidence at the tested widths. Planar excursions conditioned to reach a distant boundary exhibit transverse ranges proportional to the exit radius, while their repeated-visit fractions vary strongly with mirror density. Endpoint-label collision probabilities substantially exceed a general gluing lower bound, although they account for only part of the observed same-column crossing probability. Finally, rerouting at repeated vertices preserves the exit while substantially shortening the sampled paths. These findings characterize finite-volume transport and excursion geometry without determining an asymptotic scaling exponent or resolving infinite-plane localization.
\end{abstract}
\maketitle

\section{Introduction}\label{sec:intro}

The Lorentz mirror model combines a random environment with a completely
deterministic trajectory. Once each scatterer has been chosen, a ray
follows its local routing rule on every subsequent visit. This distinction
from a random walk is essential: repeated visits do not provide new
randomness. The model was introduced as a deterministic lattice gas by
Ruijgrok and Cohen~\cite{RuijgrokCohen}. Numerical work by Ziff, Kong, and
Cohen~\cite{ZiffKongCohen} connected the resulting paths with percolation
hulls and kinetic growth walks, and Bunimovich and
Troubetzkoy~\cite{BunimovichTroubetzkoy} studied recurrence for several
lattice-gas rules.

This paper asks three quantitative questions. How long is the penetration
scale of a finite-width cylinder? What geometry is typical of a planar
excursion conditioned to reach a distant boundary? How much probability
is retained by an endpoint-collision construction used to join two
crossings? We address these questions through reproducible finite-volume
simulations, and use elementary routing identities to interpret the data.
The numerical observables, their uncertainty, and their limitations are
the principal results.

Our computations connect several quantities that are often discussed
separately. Cylinder transmission gives a statistical confinement
envelope that remains valid beyond the simulated length. Repeated
arrivals are simultaneously a trajectory statistic and the cycle rank
of a planar trace. Endpoint-label collisions measure the loss in a
particular crossing witness. Finally, a resampling experiment separates
the geometric complexity of a trace from the length of its actual
source-to-exit component. These connections are more informative here
than fitting a single transport exponent.

\subsection{Relation to earlier work}

Kozma and Sidoravicius~\cite{KS} established a lower bound of order
\(R^{-1}\) for Lorentz escape probabilities. Kraemer and
Sanders~\cite{KraemerSanders} studied the density of paths reaching finite
boundaries using an argument with numerical inputs. On even cylinders,
Li~\cite{LiCylinder} obtained a polynomial localization scale
\(O(n^{10})\) for circumference \(2n\). Ryan~\cite{Ryan} obtained a
different estimate in the dilute regime \(p\le C/n\), with scale
\(p^{-2}\). Our density parameters are fixed, and our observed
penetration quantiles are finite-width statistics. They are neither
estimates of the unspecified constants in Li's upper bound nor a test
of Ryan's dilute asymptotic regime.

There is also an extensive physical literature on loop models with
crossings. Martins, Nienhuis, and Rietman~\cite{MartinsNienhuisRietman}
studied an integrable intersecting-loop family; Jacobsen, Read, and
Saleur~\cite{JacobsenReadSaleur} developed its Goldstone-phase
interpretation. Nahum, Serna, Somoza, and Ortu\~no~\cite{NahumEtAl}
combined field theory with simulations reaching linear sizes of order
\(10^6\), emphasizing logarithmic corrections. Their results already
show why a modest-range power fit is an unreliable way to identify
asymptotic behavior. Our considerably smaller simulations concern
boundary penetration, conditional excursion geometry, and the efficiency
of specified crossing witnesses. We do not present them as a new
determination of Goldstone-phase exponents.

In finite volume, the independent Lorentz law has loop weight one:
forming an additional closed component does not reweight an environment.
This differs from the loop-weight-two model related to the XXZ chain
by Ryan~\cite{RyanXXZ}. The hierarchical model of Lefevere and
Tasaki~\cite{LefevereTasaki} also changes the geometry and admits a
different, recursive transport analysis. These models provide useful
context but no probability input for the computations below.

The Manhattan pinball model restricts the mirror orientation using the
vertex parity. Its relation to quantum networks was developed by
Beamond, Cardy, and Chalker~\cite{BeamondCardyChalker} and by Beamond,
Owczarek, and Cardy~\cite{BeamondOwczarekCardy}. Li~\cite{LiManhattan}
proved exponential planar confinement for densities extending below
\(1/2\). Our simulations use the symmetric Lorentz law throughout;
the numerical values below are not Manhattan-model data.
These questions also belong to the broader study of localization in
random media initiated by Anderson \cite{Anderson} and developed in the
scaling theory of Abrahams, Anderson, Licciardello, and Ramakrishnan
\cite{AALR}. The Schr\"odinger results concern spectral
localization, whereas the estimates below concern paths in a local
pairing system (recent work on Schr\"odinger model for example \cite{DingSmart,LiBernoulli2D,LiZhangBernoulli3D}).
\subsection{Main numerical findings and organization}

Section~\ref{sec:methods} specifies the model, sampling, and uncertainty
calculations. Section~\ref{sec:cyl} reports cylinder penetration and the
resulting confidence envelopes. Section~\ref{sec:plane} reports planar
escape and conditional geometry. Section~\ref{sec:endpoints} measures
endpoint-collision losses, and Section~\ref{sec:intervention} reports
the resampling experiment. Section~\ref{sec:discussion} assesses what
these finite-size observations establish. The proofs needed to connect
the measurements are collected in Appendix~\ref{app:proofs}, followed
by reproduction details. In particular, the exposition does not require
the reader to follow a localization proof before seeing the data.

\section{Model, observables, and numerical protocol}\label{sec:methods}

\subsection{Local law and trajectory convention}

At each vertex, independently, the four physical half-edges
\(N,E,S,W\) are paired according to
\begin{equation}\label{eq:states}
 \begin{array}{c|c|c}
  \text{vacancy}&\text{first mirror}&\text{second mirror}\\
  (N,S)(E,W)&(N,W)(S,E)&(N,E)(S,W)\\
  1-p&p/2&p/2
 \end{array}
\end{equation}
The ray follows the matching on each arrival. The computations use
\(p\in\{1/5,2/5,2/3,9/10,1\}\); \(p=1\) is a fully occupied control.
All queried site states are stored and reused. There is no resampling
on a repeated visit in the baseline simulations.

On the plane the incoming edge is \(((-1,0),(0,0))\), with
\(X_0=(0,0)\); the origin has an unconditioned random state. Set
\begin{equation}\label{eq:planarobs}
 \tau_R=\inf\{t:\|X_t\|_\infty\ge R\},\qquad
 E_R=\{\tau_R<\infty\},\qquad
 N_R=|\{X_0,\ldots,X_{\tau_R}\}|.
\end{equation}
On \(E_R\), define
\begin{equation}\label{eq:geometryobs}
 K_R=\tau_R+1-N_R,\qquad
 D_R=\min_{i=1,2}\left(\max_{t\le\tau_R}X_t^{(i)}
                  -\min_{t\le\tau_R}X_t^{(i)}\right),\qquad
 \rho_R=\frac{K_R}{\tau_R+1}.
\end{equation}
Thus \(D_R\) measures the narrower coordinate range of the complete
excursion, not the transverse coordinate of its endpoint. The
conditional mean of \(\rho_R\) is an average of pathwise fractions,
not a ratio of two ensemble averages.

For a cylinder of circumference \(W\), let
\(S_{W,L}=\{1,\ldots,L\}\times(\Z/W\Z)\), with open left and right
boundaries and periodic transverse coordinate. Write
\begin{equation}\label{eq:A}
 A_W(L)=\Pp(\text{some left boundary port is paired to a right port in }S_{W,L}).
\end{equation}
The word ``some'' is part of the observable: \(A_W(L)\) is not the
transmission probability of one uniformly chosen incoming port.
In a half-cylinder, let \(H_W\) be the largest axial coordinate reached
by any trajectory entering from its left boundary. Cutting at the
first crossing of a given section gives
\begin{equation}\label{eq:H}
 A_W(L)=\Pp(H_W>L),\qquad
 \xi_q(W)=\inf\{L:\Pp(H_W\le L)\ge q\}.
\end{equation}
We call \(\xi_q\) a penetration quantile, to distinguish it from an
asymptotic exponential decay length.

\subsection{Experiment sizes and stopping rules}

\begin{table}[tb]
\centering\small
\caption{Sampling design. Each row counts independent environments;
different radii or lengths within one environment are correlated.}
\label{tab:design}
\begin{tabular}{@{}p{.19\textwidth}p{.42\textwidth}r@{}}
\toprule
Experiment & Densities and sizes & Environments\\
\midrule
Cylinder grid & \(p=2/5,2/3,1\); \(W=8,16,32,64,128\);
20,000 per pair & 300,000\\
Cylinder extension & \(p=2/5\); \(W=256,512\);
5,000 per width & 10,000\\
Plane & Five densities; 20,000 per density;
\(R=16,32,64,128,256,512\) & 100,000\\
Endpoint labels & \(p=2/5,2/3,1\); \(m=4h\),
\(h=4,8,16,32\); 50,000 per pair & 600,000\\
\midrule
Total & & 1,010,000\\
\bottomrule
\end{tabular}
\end{table}

Table~\ref{tab:design} gives the complete design. Cylinder explorations
trace unmatched left ports until all return to the left, or until one
reaches the right boundary at \(64W+1\). A returning component pairs
two left ports, so the second need not be traced again. This produces
\(\min(H_W,64W+1)\) and a censoring flag. All transmission indicators
for \(L\le64W\) are determined exactly by this observation, even if
the maximum is censored. In the reported data the total number of
censored cylinder samples is \(\CylinderCensored\).

The planar exploration records each nested exit radius on its first
hit and terminates at \(R=512\) or on return to its initial incoming
directed state. In the latter case the trajectory is periodic and
cannot later reach a new radius. A trajectory that has not exited
a finite box can visit only finitely many directed states; the code
checks this deterministic bound. No time limit is used to classify
a trajectory as trapped. Site arrays use generation stamps to avoid
clearing an entire box between independent environments.

Rational state probabilities are sampled using integer rejection from
\texttt{mt19937\_64}. For \(p=a/b\), an integer uniform on
\(\{0,\ldots,2b-1\}\) is assigned to the three states in groups of
sizes \(2(b-a),a,a\). Thus there is no floating-point threshold
approximation to \eqref{eq:states}. A distinct recorded seed identifies
each parameter block. The seed formulas and commands are given in
Appendix~\ref{app:reproduce}.

\subsection{Uncertainty and verification}

Cylinder and planar escape plots show pointwise 95\% Wilson intervals. Quantile
intervals use binomial order-statistic bounds. Means of conditional
geometry observables use their empirical standard errors over escaping
environments. Slope intervals for cylinder quantiles use 500 bootstrap
replicates, resampling complete environments through their histograms.
These are Monte Carlo uncertainties; they do not include finite-size
bias. Points obtained from different radii of the same trajectory are
not treated as independent measurements.

For the confinement envelopes we instead use one-sided
Clopper--Pearson upper limits with error allocation \(0.01/17\) to
each density--width pair. The fixed probe is \(L=16W\) in every pair.
This gives a simultaneous 99\% confidence statement over the 17
probabilities, including the cylinder extension. It is a statistical
statement about random sampling, not an exact enumeration of their
probabilities.

An independent reference implementation constructs connected components
of a graph of physical edges, rather than following ray-update formulas.
It agrees with the production routing code on all 19,683 configurations
of a \(3\times3\) box and all three bottom ports, giving 59,049
comparisons. A further 3,000 odd-cylinder samples check the parity
identity for transmitting ports. The raw planar records satisfy
\(N_R+K_R=\tau_R+1\), and an independently reconstructed example trace
has cycle rank \(K_R\). These checks test the boundary conventions and
state reuse that are most consequential for the experiment.

\section{Cylinder penetration and confinement envelopes}\label{sec:cyl}

\subsection{Transmission curves and penetration quantiles}

Figure~\ref{fig:cyl-tail} shows the transmission probability against
\(L/W\). At a given circumference, transmission decreases strongly
with length. Vacancies substantially increase the penetration scale:
for \(W=128\), the 90\% penetration quantiles at
\(p=2/5,2/3,1\) are respectively \(1479,805,234\).
The event at the left of \eqref{eq:A} includes all incoming ports,
so this difference is not a consequence of choosing a favorable
single initial height.

\begin{figure}[tb]
\centering\includegraphics[width=\textwidth]{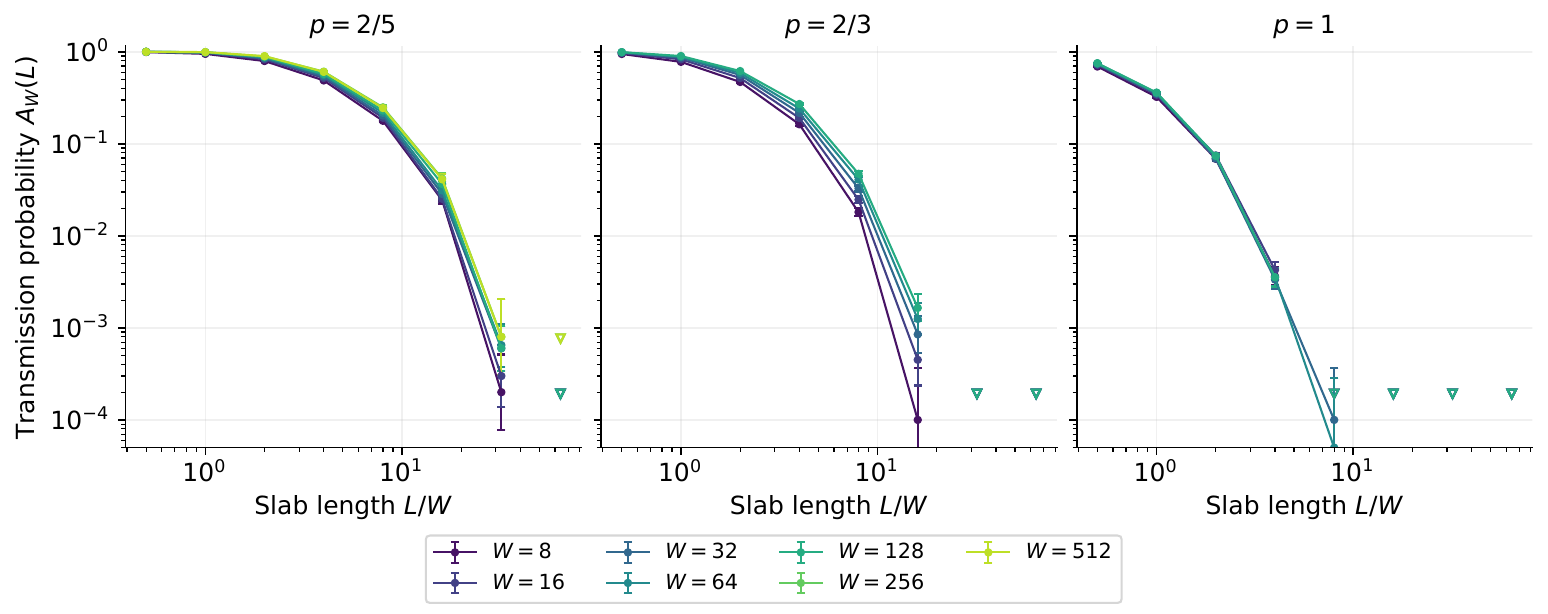}
\caption{Cylinder transmission. Circles show nonzero Monte Carlo
estimates with pointwise 95\% Wilson intervals; open downward triangles
mark upper interval endpoints for zero observed successes. No zero
count is interpreted as zero probability. Curves connect measured
lengths and are not fitted tail laws.}
\label{fig:cyl-tail}
\end{figure}

Table~\ref{tab:cyl} reports the measured quantiles and the counts used
for the confidence envelope. The wider \(p=2/5\) extension gives
\(\xi_{.9}(512)=\WideNinety\), with 95\% interval
\([\WideNinetyLow,\WideNinetyHigh]\). Its normalized value is
\(\WideNinetyScaled\), compared with \(83/8=10.375\) at \(W=8\).

\begin{table}[tb]
\centering\small
\caption{Cylinder penetration quantiles and the fixed probe \(L=16W\).
The main grid uses 20,000 environments per row; the two rows with
\(W>128\) use 5,000. Brackets give pointwise 95\% quantile intervals.
\(k_L\) is the number transmitting at \(16W\), and \(u\) is the
one-sided upper limit with simultaneous 99\% coverage across all rows,
rounded upward to five decimal places.}
\label{tab:cyl}
\begin{tabular}{@{}rrrcrr@{}}
\toprule
\(p\) & \(W\) & \(\xi_{.5}\) & \(\xi_{.9}\) [95\% CI] & \(k_L\) & \(u\)\\
\midrule
$2/5$ & 8 & 32 & 83 [82, 84] & 489 & 0.02820 \\
$2/5$ & 16 & 67 & 170 [168, 173] & 506 & 0.02911 \\
$2/5$ & 32 & 138 & 347 [341, 353] & 571 & 0.03258 \\
$2/5$ & 64 & 289 & 721 [711, 732] & 624 & 0.03539 \\
$2/5$ & 128 & 593 & 1479 [1454, 1504] & 733 & 0.04116 \\
$2/5$ & 256 & 1251 & 3072 [2970, 3167] & 210 & 0.05200 \\
$2/5$ & 512 & 2511 & 6149 [5967, 6330] & 206 & 0.05112 \\
$2/3$ & 8 & 16 & 40 [39, 40] & 2 & 0.00060 \\
$2/3$ & 16 & 34 & 85 [84, 87] & 9 & 0.00118 \\
$2/3$ & 32 & 72 & 181 [178, 184] & 17 & 0.00175 \\
$2/3$ & 64 & 152 & 382 [376, 388] & 25 & 0.00229 \\
$2/3$ & 128 & 324 & 805 [791, 816] & 33 & 0.00281 \\
$1$ & 8 & 6 & 15 [14, 15] & 0 & 0.00038 \\
$1$ & 16 & 12 & 29 [29, 30] & 0 & 0.00038 \\
$1$ & 32 & 25 & 57 [56, 58] & 0 & 0.00038 \\
$1$ & 64 & 50 & 116 [114, 118] & 0 & 0.00038 \\
$1$ & 128 & 101 & 234 [230, 237] & 0 & 0.00038 \\
\bottomrule
\end{tabular}
\end{table}

The normalization in Figure~\ref{fig:cyl-scale} makes finite-size drift
visible. The fully occupied control is close to a constant after
division by \(W\), while the two vacancy densities show an increase.
Division by \(W\log W\) produces a downward drift over the measured
range. Consequently neither plot establishes an asymptotic scaling
law: a slowly varying correction, including a logarithm with a
non-negligible additive constant, remains unresolved.

\begin{figure}[tb]
\centering\includegraphics[width=.94\textwidth]{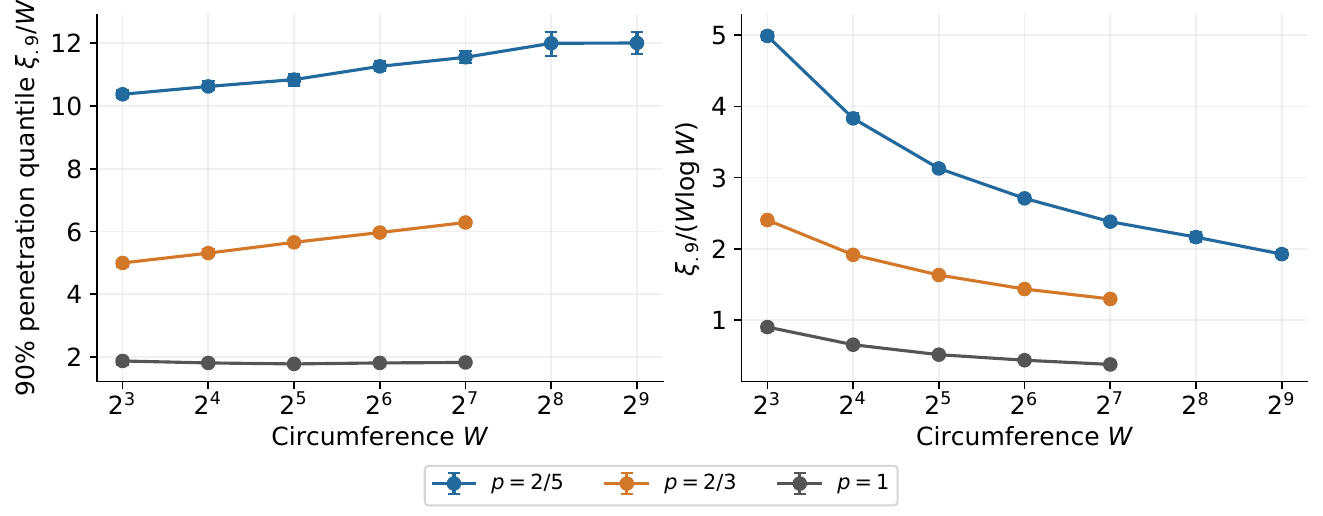}
\caption{Two normalizations of the 90\% penetration quantile. Error
bars are binomial order-statistic 95\% intervals. The drift under both
normalizations is retained rather than removed by fitting a collapse.}
\label{fig:cyl-scale}
\end{figure}

For a compact description of the data, Table~\ref{tab:fits} gives
unweighted least-squares slopes of \(\log\xi_{.9}\) against
\(\log W\). Changing the lower end of the window is a finite-size
sensitivity check. The bootstrap intervals describe sampling error
conditional on that fitting rule, not confidence intervals for a
universal exponent. In particular, a slope near one over these
widths does not prove a linear localization length.

\begin{table}[tb]
\centering\small
\caption{Descriptive power slopes for \(\xi_{.9}\). The last column
is a 95\% percentile interval from 500 environment-level bootstrap
replicates. Different windows for the same density reuse data.}
\label{tab:fits}
\begin{tabular}{@{}rrrr@{}}
\toprule
\(p\) & Width window & Slope & Bootstrap interval\\
\midrule
$2/5$ & 8--512 & 1.038 & [1.032, 1.045] \\
$2/5$ & 32--512 & 1.039 & [1.027, 1.049] \\
$2/5$ & 128--512 & 1.028 & [1.005, 1.051] \\
$2/3$ & 8--128 & 1.083 & [1.077, 1.093] \\
$2/3$ & 16--128 & 1.081 & [1.069, 1.092] \\
$2/3$ & 32--128 & 1.076 & [1.059, 1.091] \\
$1$ & 8--128 & 0.993 & [0.984, 1.013] \\
$1$ & 16--128 & 1.006 & [0.985, 1.013] \\
$1$ & 32--128 & 1.019 & [1.000, 1.034] \\
\bottomrule
\end{tabular}
\end{table}

\subsection{From a finite measurement to a distance-uniform envelope}

Let \(B_W(M)\) be the event that all trajectories through the section
\(x=0\), with all incoming directions, stay in \([-M,M]\times(\Z/W\Z)\).
The elementary independent-slab bound proved in
Appendix~\ref{app:slab} is
\begin{equation}\label{eq:envelope}
 \Pp(B_W(kL)^c)\le 2 A_W(L)^k,\qquad k=1,2,\ldots.
\end{equation}
Thus each measured upper limit \(u\) in Table~\ref{tab:cyl} supplies
the envelope \(\min(1,2u^k)\) at every distance \(16kW\).
With simultaneous 99\% Monte Carlo confidence, all 17 resulting
envelopes hold for all positive integers \(k\). There is no additional
multiple-comparison cost over \(k\): once the one-slab inequalities
hold, \eqref{eq:envelope} is deterministic in those probabilities.

The largest upper limit in the table is \(\WorstUpper\), well below
one. This gives useful quantitative information for each tested
width. It does not justify interpolation to an untested width or
passage to \(W\to\infty\). Likewise, zero transmitting samples at
\(p=1,L=16W\) yield a positive upper confidence limit, as shown
explicitly in the table.

This finite-width conclusion concerns a stronger initial condition
than a single selected ray, but it does not resolve the planar
problem. The rigorous polynomial cylinder estimate of
\cite{LiCylinder} and the numerical envelopes answer different
questions: one controls all sufficiently large widths through a
theorem, while the other gives evaluated constants at a finite
list of widths.

\section{Planar escape and conditional excursion geometry}\label{sec:plane}

\subsection{Escape across nested scales}

The plane experiments use fresh independent environments for each
trajectory, recording six nested exit events in that environment.
Figure~\ref{fig:escape} shows both the escape probabilities and the
local doubling slope
\begin{equation}\label{eq:beta}
 \beta_{mathrm{eff}}(R)=
 -\frac{\log[\Pp(E_{2R})/\Pp(E_R)]}{\log 2}.
\end{equation}
To respect the nesting, its uncertainty is obtained from the Wilson
interval for \(\Pp(E_{2R}\mid E_R)\), then transformed by
\(-\log_2\). Treating the two escape estimates as independent would
misrepresent this uncertainty.

\begin{figure}[tb]
\centering\includegraphics[width=.94\textwidth]{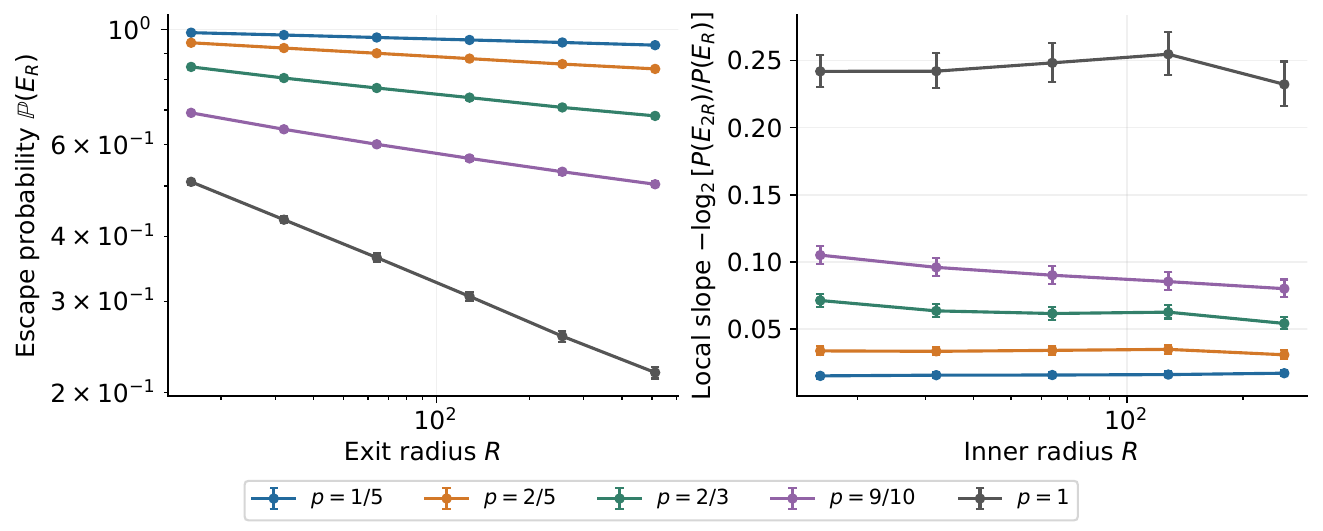}
\caption{Planar escape probabilities and local doubling slopes.
Every density uses 20,000 independent environments; the nested
radii within a density use the same trajectories. Error bars in
the right panel account for this nesting through conditional
binomial intervals. The curves are finite-size diagnostics.}
\label{fig:escape}
\end{figure}

At \(R=512\), the observed escape probabilities range from
\(0.93305\) at \(p=1/5\) to \(0.21875\) at \(p=1\)
(Table~\ref{tab:plane}). The persistence of a substantial escape
probability over this range is compatible with eventual periodicity:
it is not evidence of an atom at infinite trajectory length.
The local slopes avoid imposing a power law on the full dataset,
but they also cannot distinguish a slowly decaying function from
a nonzero limit at the available radii.

\begin{table}[tb]
\centering\small
\caption{Planar results at \(R=512\). Each density has 20,000 trials.
\(n_E\) is the number reaching the boundary; all geometry columns
condition on this event. Brackets are 95\% Wilson intervals for
escape. Full uncertainty intervals for the geometry and medians
are provided in the machine-readable summary.}
\label{tab:plane}
\begin{tabular}{@{}rrcrrr@{}}
\toprule
\(p\) & \(n_E\) & \(\widehat{\Pp}(E_R)\) [95\% CI] &
\(\overline{D_R/R}\) & \(\overline{\rho_R}\) & Median \(\tau_R\)\\
\midrule
$1/5$ & 18661 & 0.9331 [0.9295, 0.9364] & 0.923 & 0.106 & 24778 \\
$2/5$ & 16797 & 0.8398 [0.8347, 0.8449] & 0.915 & 0.218 & 45227 \\
$2/3$ & 13640 & 0.6820 [0.6755, 0.6884] & 0.908 & 0.320 & 63720 \\
$9/10$ & 10075 & 0.5038 [0.4968, 0.5107] & 0.893 & 0.356 & 68877 \\
$1$ & 4375 & 0.2188 [0.2131, 0.2245] & 0.871 & 0.331 & 62760 \\
\bottomrule
\end{tabular}
\end{table}

\subsection{Transverse spread and repeated visits}

Conditioned on reaching \(R=512\), the mean of \(D_R/R\) lies
between \(0.871\) and \(0.923\) across the five densities.
Figure~\ref{fig:geometry} shows that a macroscopic transverse range
is already visible over the measured scales. These observations are
much broader than what a small-power lower bound on width would
require. They support studying the full distribution of \(D_R/R\),
while leaving its limiting law open.

\begin{figure}[tb]
\centering\includegraphics[width=.94\textwidth]{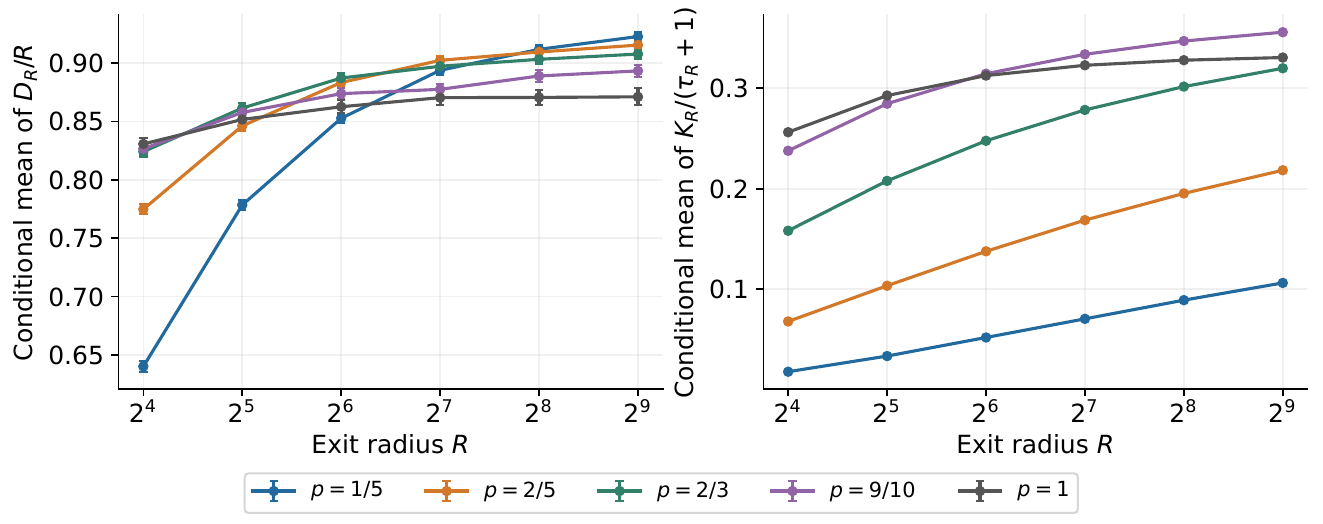}
\caption{Conditional geometry of escaping planar excursions.
Left: mean narrower coordinate range divided by exit radius.
Right: mean pathwise repeated-arrival fraction. Error bars are
95\% intervals from the empirical standard errors among escaping
environments at each radius.}
\label{fig:geometry}
\end{figure}

Repeated visits show a different dependence on density. At
\(R=512,p=2/5\), the mean of \(\rho_R\) is \(0.21838\),
with 95\% interval approximately \([0.21796,0.21879]\).
The corresponding mean \(K_R/R\) is \(24.59\), whereas the
elementary repeated-arrival bound in Appendix~\ref{app:fresh}
has coefficient \(p^3/96=1/1500\). The large gap is useful
information about that estimate's conservatism; it does not
by itself improve the probability inequality.

The repeated fraction is not monotone over the five tested densities:
it is about \(0.356\) at \(p=9/10\) and \(0.331\) at \(p=1\).
Thus a simple monotone dependence on the mirror density would not
describe this conditional observable at \(R=512\). The conditioning
and the presence of crossing passages both matter. We do not
infer a parameter-monotonicity theorem from this finite comparison.

For each escaping excursion, cut its initial incoming edge at its
midpoint and include the resulting source half-edge in the physical
trace \(G_R\). Appendix~\ref{app:trace} proves the exact identity
\begin{equation}\label{eq:cycle-rank}
 \beta_1(G_R)=K_R.
\end{equation}
The same counts therefore measure trace cycle rank. This identity
does not count vertex-disjoint short cycles or closed dynamical
orbits. No such interpretation is attached to the numerical
values of \(K_R\).

\section{Endpoint collisions and the cost of a crossing witness}\label{sec:endpoints}

\subsection{A measurable refinement of a generic bound}

Consider the planar rectangle
\(Q_{m,h}=\{1,\ldots,m\}\times\{1,\ldots,h\}\), open on all
four sides. Let \(\Cross(m,h)\) mean that some bottom boundary
port connects to the top, and let \(\Col(m,h)\) mean that such
a connection has the same entry and exit column. On a crossing
configuration, select the smallest successful bottom column
\(i\); its top exit \(j\) is unique. Write
\[
 q=\Pp(\Cross(m,h)),\qquad
 r_{ij}=\Pp(\text{selected label is }(i,j)),\qquad
 S=\sum_{i,j}r_{ij}^2.
\]
The selection rule is fixed before inspecting the sample.
Reflection of an independent second rectangle gives
\begin{equation}\label{eq:glue-main}
 \Pp(\Col(m,2h))\ge S\ge q^2/m^2.
\end{equation}
Appendix~\ref{app:collision} proves the inclusions and the inequality.
We measure both losses in this chain, rather than treating its
rightmost term as an estimate of the actual same-column probability.

Given \(N\) sampled environments, let \(c_{ij}\) count selected
labels, excluding noncrossing configurations. The collision
estimator is
\begin{equation}\label{eq:collision-estimator}
 \widehat S=\frac{\sum_{i,j}c_{ij}(c_{ij}-1)}{N(N-1)},\qquad
 \widehat G=\frac{m^2\widehat S}{\widehat q^2}.
\end{equation}
The first estimator is unbiased: it counts matching labels among
distinct sampled environments. Squaring empirical frequencies
would add a diagonal term of order \(1/N\), which is appreciable
at the larger sizes. The ratio \(\widehat G\) is a plug-in
estimate of the collision gain; it is not claimed to be unbiased.
Its intervals use the joint influence of \(\widehat S\) and
\(\widehat q\) from the same environments.

\subsection{Observed collision gains}

Table~\ref{tab:endpoint} and Figure~\ref{fig:endpoint} show the
results for \(m=4h\). The measured gains range from about
\(4.7\) to \(12.3\), so the uniform distribution on \(m^2\)
labels is a poor quantitative description of this selector.
This gain depends on the selection rule and rectangle geometry.
It does not imply the same gain for a different witness, a
different aspect ratio, or a conditional boundary law.

\begin{table}[tb]
\centering\small
\caption{Endpoint statistics for 50,000 environments per row,
with \(h=m/4\). Brackets give approximate pointwise 95\%
intervals for the collision gain. The last column is the
separately measured event in the full \(m\times2h\) rectangle.}
\label{tab:endpoint}
\begin{tabular}{@{}rrrccr@{}}
\toprule
\(p\) & \(m\) & \(\widehat q\) & \(\widehat S\) &
\(\widehat G\) [95\% CI] & \(\widehat{\Pp}(\Col)\)\\
\midrule
$2/5$ & 16 & 0.9999 & 0.047810 & 12.24 [12.03, 12.45] & 0.40602 \\
$2/5$ & 32 & 1.0000 & 0.009209 & 9.43 [9.31, 9.55] & 0.19900 \\
$2/5$ & 64 & 1.0000 & 0.002348 & 9.62 [9.51, 9.72] & 0.11236 \\
$2/5$ & 128 & 1.0000 & 0.000651 & 10.67 [10.55, 10.79] & 0.06310 \\
$2/3$ & 16 & 0.9963 & 0.018215 & 4.70 [4.65, 4.75] & 0.18464 \\
$2/3$ & 32 & 0.9985 & 0.005477 & 5.63 [5.56, 5.69] & 0.11368 \\
$2/3$ & 64 & 0.9993 & 0.001573 & 6.45 [6.38, 6.52] & 0.06536 \\
$2/3$ & 128 & 0.9997 & 0.000446 & 7.31 [7.23, 7.40] & 0.03642 \\
$1$ & 16 & 0.9381 & 0.020607 & 5.99 [5.93, 6.06] & 0.17168 \\
$1$ & 32 & 0.9476 & 0.006117 & 6.98 [6.89, 7.06] & 0.10174 \\
$1$ & 64 & 0.9518 & 0.001635 & 7.39 [7.30, 7.48] & 0.05502 \\
$1$ & 128 & 0.9544 & 0.000429 & 7.72 [7.63, 7.82] & 0.02762 \\
\bottomrule
\end{tabular}
\end{table}

\begin{figure}[tb]
\centering\includegraphics[width=.94\textwidth]{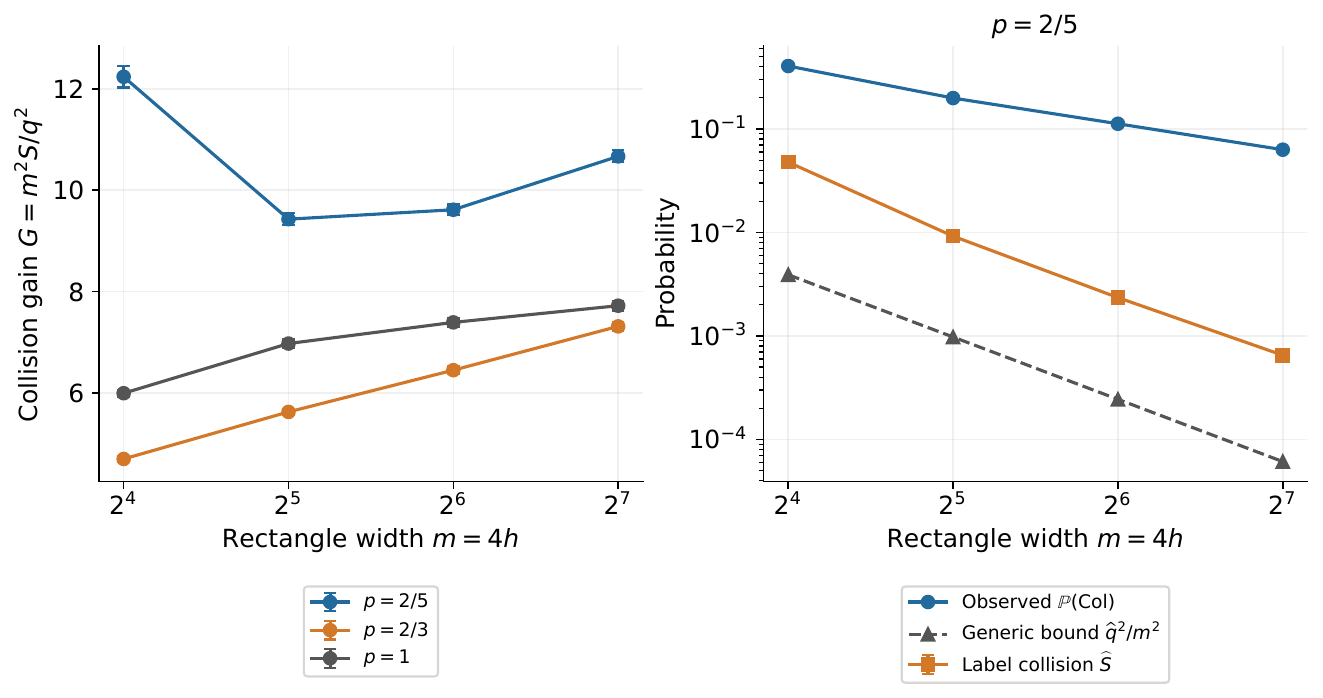}
\caption{Left: gain over the uniform-label Cauchy--Schwarz bound.
Right: the three levels of \eqref{eq:glue-main} for \(p=2/5\).
Collision error bars use delete-one jackknife standard errors;
gain intervals also account for estimation of \(q\). The top
curve is an event probability, while the two lower curves are
witness-based lower bounds on it.}
\label{fig:endpoint}
\end{figure}

For example, at \(p=2/5,m=128,h=32\),
\[
 \widehat{\Pp}(\Col)=0.06310,\qquad
 \widehat S=0.0006513,\qquad
 \widehat q^2/m^2=0.00006104.
\]
The measured collision gain is \(10.67\), but the actual
same-column event is still about 97 times as frequent as the
selected matching-label witness. This separates two sources
of loss: nonuniform endpoint probabilities account for roughly
one order of magnitude, while restricting attention to one
chosen crossing in each rectangle discards substantially more.
The latter comparison points toward richer collections of
crossing witnesses as a computational target, without asserting
that their overlaps can already be controlled analytically.

The aspect ratio here is four. In particular, these measurements
are not substituted into a theorem requiring rectangles with
aspect ratio 100. Also, \(\widehat q=1\) in a row means that
every sampled rectangle crossed, not that its crossing
probability is exactly one.

\section{Resampling repeated vertices}\label{sec:intervention}

Graph cycles can be abundant even though the actual escaping
trajectory contains no closed dynamical component. To measure
this distinction, we take the first 1,000 independent
\(p=2/5\) trajectories in the planar experiment that reach
\(R=64\). In each stopped trace, all vertices visited twice
are simultaneously assigned new independent states from
\eqref{eq:states}; all other states are retained. The resampling
uses a random stream separate from the baseline exploration.
Although a longer original trajectory was recorded for the
planar study, the intervention set uses only visits before
the first exit at \(R=64\).

\begin{figure}[tb]
\centering\includegraphics[width=.94\textwidth]{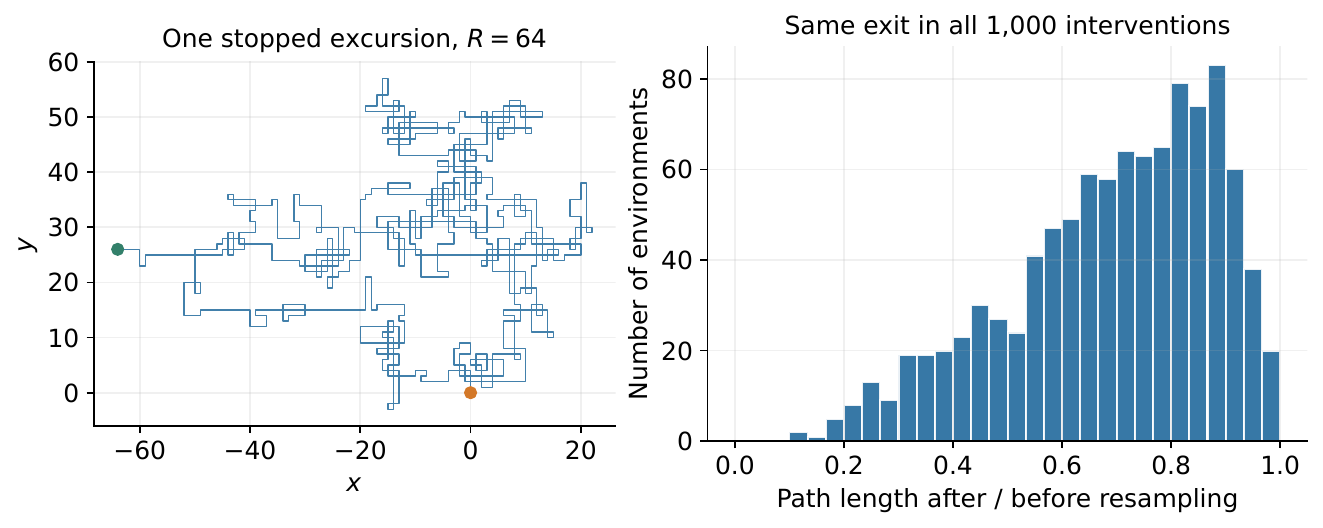}
\caption{An escaping trace and the resampling intervention.
Left: the first accepted \(p=2/5,R=64\) trace, shown only as
an illustration, with source and exit marked. Right: path-length
retention across 1,000 independently sampled escaping
environments, one resampling per trace.}
\label{fig:rerouting}
\end{figure}

The exit vertex is unchanged in all 1,000 interventions.
This is predicted exactly by the two-terminal pairing argument
in Appendix~\ref{app:resampling}; it is also checked by the code.
The nontrivial numerical statistic is the amount of trace
retained in the new source component. The mean ratio of new
to original path length is \(0.6914\), with approximate
95\% interval \([0.6796,0.7032]\). Its median is \(0.7246\),
and its empirical 10th and 90th percentiles are \(0.4045\)
and \(0.9074\). Of the 1,000 paths, 994 shorten.

Thus the intervention removes about 31\% of the path length
on average while preserving escape to the same boundary point.
It is a conditional intervention on an observed trace, not a
new independent sample of the Lorentz environment conditioned
on escape. Its purpose is to quantify how much geometric trace
can be separated into closed components without affecting
the source-to-exit connection. In particular, the observed
repeated vertices cannot be interpreted as independent
opportunities to destroy escape by local state changes.

\section{Discussion}\label{sec:discussion}

The cylinder data show a nearly linear penetration scale over
the measured widths, together with density-dependent finite-size
drift. The most directly usable output is the table of evaluated
one-slab confidence limits: these imply distance-uniform
confinement envelopes at the tested widths without a fitted
tail model. The power slopes are descriptive summaries of
finite data and do not replace that probability statement.

For the plane, large transverse range and substantial repeated
visitation coexist with sizable escape probabilities. The
trace-rank identity makes the revisit statistic geometrically
meaningful, while the intervention shows why high cycle rank
need not translate into independent closing mechanisms.
Endpoint experiments likewise distinguish an actual crossing
event from a selected witness for it. Improving the endpoint
distribution estimate alone leaves a much larger selection
loss in the measured examples.

Three limitations govern interpretation. The confidence
envelopes concern only the finite list of sampled widths.
The planar radii do not distinguish logarithmic decay from
other very slow asymptotics. Finally, collision and resampling
statistics depend on their stated selection and conditioning
rules. These limitations are part of the numerical conclusions,
not assumptions to be supplied by an unproved localization
argument. Larger widths, distributions of normalized excursion
shapes, and several simultaneous crossing witnesses are natural
next measurements, but no conclusion here depends on their
outcome.

\section*{Data and code availability}
The accompanying reproduction archive contains the complete
per-environment records, summary tables, plotting and analysis
code, the independent graph verifier, and the LaTeX source.
All reported numbers and figures are generated from those
records. The archive also records the seeds, commands,
software versions, and verification results.

\section*{Acknowledgement}
OpenAI's GPT was used in preparing the manuscript, developing
and checking simulation and analysis code, and organizing the
supporting mathematical arguments.

\appendix
\section{Supporting routing arguments}\label{app:proofs}

\subsection{Boundary pairings and independent slabs}\label{app:slab}

\begin{lemma}\label{lem:boundary}
In a finite domain with fixed local pairings, every incoming
boundary port is connected to exactly one other boundary
port. Boundary ports are paired perfectly. A path between
ports traverses each physical edge at most once.
\end{lemma}
\begin{proof}
Construct a graph whose vertices are physical edges, including
edges cut at the domain boundary. For each local pair of
half-edges, join the corresponding edge-vertices. Internal
edge-vertices have degree two and boundary edge-vertices
have degree one. A component containing a boundary vertex
is a finite path, with exactly two degree-one endpoints;
the other components are cycles. This proves the assertions.
\end{proof}

For a cylinder with \(W\) left ports, the number of transmitting
ports is congruent to \(W\) modulo two: all remaining left
ports are paired among themselves. This explains both the
odd-width control and the use of even widths in the
confinement experiment.

\begin{proposition}\label{prop:slab}
For positive integers \(W,L,k\), the cylinder events defined
in the text satisfy \eqref{eq:envelope}.
\end{proposition}
\begin{proof}
If a trajectory through \(x=0\) leaves through
\(x=kL+1\), it traverses each of the disjoint slabs
\(\{jL+1,\ldots,(j+1)L\}\times(\Z/W\Z)\),
\(j=0,\ldots,k-1\). For each slab, take the last left
entry preceding its first right exit; the intervening
segment is a slab crossing. Each crossing event depends
only on the states in that slab. Their independence
bounds the probability by \(A_W(L)^k\).
This implication already includes all initial ports and
directions, so no additional factor in \(W\) is needed.
Reflection gives the same bound for leaving on the left.
A union bound proves the claim.
\end{proof}

For completeness, if a binomial count is \(k_0\) out of
\(N\) trials, its one-sided upper confidence limit at
error \(\eta\) is
\[
 u=\operatorname{Beta}^{-1}(1-\eta;k_0+1,N-k_0)
\]
when \(k_0<N\), and \(u=1\) otherwise. With
\(\eta=0.01/17\), the union bound gives the coverage
used in Section~\ref{sec:cyl}. It does not require
independence between the 17 estimates.

\subsection{Endpoint-collision gluing}\label{app:collision}

\begin{proposition}
For the symmetric Lorentz law and the fixed selector of
Section~\ref{sec:endpoints}, \eqref{eq:glue-main} holds.
\end{proposition}
\begin{proof}
Take independent configurations \(\omega,\eta\) on
\(Q_{m,h}\). Place \(\omega\) below a reflected copy
of \(\eta\), interchanging the two mirror orientations
under reflection. The resulting \(m\times2h\)
configuration has the original product law. If the two
latent selected labels both equal \((i,j)\), the lower
path goes from \(i\) to \(j\), and the reflected upper
path goes from \(j\) to \(i\). They concatenate along
their shared interface edge to give a same-column crossing.
The matching-label events are disjoint and have total
probability \(\sum r_{ij}^2=S\). Since
\(\sum r_{ij}=q\), Cauchy--Schwarz gives
\(S\ge q^2/m^2\). The selector is applied before
reflection; its equivariance is not assumed.
\end{proof}

For distinct sampled environments \(a,b\), let the
collision indicator be one only when both have the same
successful label. Its expectation is \(S\). Averaging
over the \(N(N-1)\) ordered pairs gives
\eqref{eq:collision-estimator} and proves unbiasedness.
If \(c_z\) is the count of a successful label \(z\),
the centered delete-one pseudo-value for an observation
with that label is
\[
 I_a=\frac{2\{c_z-1-(N-1)\widehat S\}}{N-2}.
\]
For an unsuccessful observation replace \(c_z-1\) by
zero. The standard deviation of these pseudo-values
divided by \(\sqrt N\) is the reported jackknife
standard error. The gain interval uses the joint
delta-method pseudo-value
\[
 \frac{m^2}{\widehat q^2}I_a
 -\frac{2m^2\widehat S}{\widehat q^3}
   (\mathbf1_{\{a\text{ crosses}\}}-\widehat q).
\]
These approximate intervals measure sampling uncertainty
and are not certified deterministic bounds on \(S\).

\subsection{A conservative repeated-arrival estimate}\label{app:fresh}

\begin{proposition}\label{prop:fresh}
For \(0<p\le1\) and \(R\ge2\),
\begin{equation}\label{eq:repeats-bound}
 \Pp(E_R,K_R<p^3R/96)\le 2e^{-p^3R/32}.
\end{equation}
\end{proposition}
\begin{proof}
Reveal a state only when its outgoing action is needed.
Exclude \(X_0\), and partition subsequent arrivals
into slots of consecutive new vertices, each ending at
the next repeated arrival or termination. Empty slots
are allowed. After termination all slot lengths are
zero. Write \(\ell_j\) for the lengths and
\(S_j=\sum_{i=1}^j\ell_i\).

Conditional on a slot's history, three consecutive
left turns, or three consecutive right turns, force
a return to the predecessor of the first new vertex,
unless the slot ends even earlier. Complete a stopped
three-action word with auxiliary independent actions
to formalize the comparison. Each word has probability
\((p/2)^3\), and they are disjoint. With
\(\delta=p^3/4\), iteration gives
\[
 \Pp(\ell_j>3r\mid\text{history at its start})
 \le(1-\delta)^r.
\]
Thus each slot is conditionally dominated by three
times a geometric random variable of success parameter
\(\delta\). At
\(\theta=\tfrac13\log(1/(1-\delta/2))\), its
conditional exponential moment is at most two.
Repeated conditioning yields \(\Ep e^{\theta S_j}\le2^j\).

Reaching radius \(R\) requires at least \(R\) new
vertices after \(X_0\). On \(E_R\), these occur
in at most \(K_R+1\) slots. If
\(K_R<aR\), where \(a=p^3/96=\delta/24\),
then \(S_m\ge R\) for \(m=\lceil aR\rceil\).
Markov's inequality gives
\[
 \Pp(E_R,K_R<aR)\le2^m e^{-\theta R}
 \le2e^{-(\theta-a\log2)R}
 \le2e^{-\delta R/8},
\]
using \(\theta\ge\delta/6\). This is
\eqref{eq:repeats-bound}. Slot independence is never
assumed.
\end{proof}

\subsection{Trace rank and resampling}\label{app:trace}

\begin{proposition}\label{prop:trace}
On \(E_R\), the cut trace \(G_R\) has two leaves,
all remaining degrees are two or four, and there are
exactly \(K_R\) degree-four vertices. Its cycle rank
and its number of bounded faces both equal \(K_R\).
\end{proposition}
\begin{proof}
Before reaching a new radius, the ray cannot repeat
a physical edge: in the degree-two pairing graph
this would close its component. Each nonterminal
visit uses two previously unused incident half-edges,
so a physical vertex is visited at most twice.
The source half-edge accounts for the initial arrival;
the exit is new and has degree one. This proves the
degree assertions, including the count of degree-four
vertices.

There are \(N_R+1\) graph vertices after adding the
source midpoint, and \(\tau_R+1=N_R+K_R\) graph
edges after adding its half-edge. Connectivity gives
\(\beta_1=|E|-|V|+1=K_R\). Euler's formula in the
physical plane gives the bounded-face count.
\end{proof}

\begin{proposition}\label{app:resampling}
Fix an escaping cut trace. Arbitrary simultaneous
changes of local states at its degree-four vertices,
with all other states fixed, preserve its source-to-exit
connection. The new path uses only edges of the old
trace and has no greater length.
\end{proposition}
\begin{proof}
At a degree-two vertex the unchanged state pairs its
two trace half-edges. At a degree-four vertex every
allowed new state pairs its four trace half-edges
among themselves. The source and exit are the only
unmatched boundary ports. By the finite pairing
argument, the source component must end at the exit;
all other components are cycles. It follows the
actual modified states and uses only the old trace
edges. No interior vertex lies on the stopping
boundary, so this is also the new first exit.
Since no edge is repeated, its length cannot exceed
the original path length.
\end{proof}


\begin{thebibliography}{99}
\bibitem{AALR}
E.~Abrahams, P.~W.~Anderson, D.~C.~Licciardello, and T.~V.~Ramakrishnan,
\emph{Scaling theory of localization: Absence of quantum diffusion in
two dimensions},
Phys. Rev. Lett. \textbf{42} (1979), no.~10, 673--676.
\url{https://doi.org/10.1103/PhysRevLett.42.673}.

\bibitem{Anderson}
P.~W.~Anderson,
\emph{Absence of diffusion in certain random lattices},
Phys. Rev. \textbf{109} (1958), no.~5, 1492--1505.
\url{https://doi.org/10.1103/PhysRev.109.1492}.

\bibitem{BeamondCardyChalker}
E.~J.~Beamond, J.~Cardy, and J.~T.~Chalker,
\emph{Quantum and classical localization, the spin quantum Hall effect,
and generalizations},
Phys. Rev. B \textbf{65} (2002), no.~21, 214301.
\url{https://doi.org/10.1103/PhysRevB.65.214301}.

\bibitem{BeamondOwczarekCardy}
E.~J.~Beamond, A.~L.~Owczarek, and J.~Cardy,
\emph{Quantum and classical localization and the Manhattan lattice},
J. Phys. A: Math. Gen. \textbf{36} (2003), no.~41, 10251--10267.
\url{https://doi.org/10.1088/0305-4470/36/41/001}.

\bibitem{BunimovichTroubetzkoy}
L.~A.~Bunimovich and S.~E.~Troubetzkoy,
\emph{Recurrence properties of Lorentz lattice gas cellular automata},
J. Stat. Phys. \textbf{67} (1992), nos.~1--2, 289--302.
\url{https://doi.org/10.1007/BF01049035}.

\bibitem{DingSmart}
J.~Ding and C.~K.~Smart,
\emph{Localization near the edge for the Anderson Bernoulli model on
the two dimensional lattice},
Invent. Math. \textbf{219} (2020), no.~2, 467--506.
\url{https://doi.org/10.1007/s00222-019-00910-4}.

\bibitem{JacobsenReadSaleur}
J.~L.~Jacobsen, N.~Read, and H.~Saleur,
\emph{Dense loops, supersymmetry, and Goldstone phases in two dimensions},
Phys. Rev. Lett. \textbf{90} (2003), no.~9, 090601.
\url{https://doi.org/10.1103/PhysRevLett.90.090601}.

\bibitem{KS}
G.~Kozma and V.~Sidoravicius,
\emph{Lower bound for the escape probability in the Lorentz mirror
model on \(\mathbb Z^2\)},
Israel J. Math. \textbf{209} (2015), no.~2, 683--685.
\url{https://doi.org/10.1007/s11856-015-1233-1}.

\bibitem{KraemerSanders}
A.~S.~Kraemer and D.~P.~Sanders,
\emph{Zero density of open paths in the Lorentz mirror model for
arbitrary mirror probability},
J. Stat. Phys. \textbf{156} (2014), no.~5, 908--916.
\url{https://doi.org/10.1007/s10955-014-1038-3}.

\bibitem{LefevereTasaki}
R.~Lefevere and H.~Tasaki,
\emph{Hierarchical Lorentz mirror model: Normal transport and a universal
\(2/3\) mean--variance law},
arXiv:2602.07988v5, 2026.
\url{https://arxiv.org/abs/2602.07988v5}.

\bibitem{LiManhattan}
L.~Li,
\emph{On the Manhattan pinball problem},
Electron. Commun. Probab. \textbf{26} (2021), article 25, 1--11.
\url{https://doi.org/10.1214/21-ECP394}.

\bibitem{LiBernoulli2D}
L.~Li,
\emph{Anderson--Bernoulli localization at large disorder on the 2D lattice},
Commun. Math. Phys. \textbf{393} (2022), no.~1, 151--214.
\url{https://doi.org/10.1007/s00220-022-04366-1}.

\bibitem{LiCylinder}
L.~Li,
\emph{Polynomial bound for the localization length of the Lorentz
mirror model on the 1D cylinder},
Electron. Commun. Probab., to appear.
Version cited: arXiv:2010.05900v3, 2026.
\url{https://arxiv.org/abs/2010.05900v3}.

\bibitem{LiZhangBernoulli3D}
L.~Li and L.~Zhang,
\emph{Anderson--Bernoulli localization on the three-dimensional lattice
and discrete unique continuation principle},
Duke Math. J. \textbf{171} (2022), no.~2, 327--415.
\url{https://doi.org/10.1215/00127094-2021-0038}.

\bibitem{MartinsNienhuisRietman}
M.~J.~Martins, B.~Nienhuis, and R.~Rietman,
\emph{Intersecting loop model as a solvable super spin chain},
Phys. Rev. Lett. \textbf{81} (1998), no.~3, 504--507.
\url{https://doi.org/10.1103/PhysRevLett.81.504}.

\bibitem{NahumEtAl}
A.~Nahum, P.~Serna, A.~M.~Somoza, and M.~Ortu\~no,
\emph{Loop models with crossings},
Phys. Rev. B \textbf{87} (2013), no.~18, 184204.
\url{https://doi.org/10.1103/PhysRevB.87.184204}.

\bibitem{RuijgrokCohen}
Th.~W.~Ruijgrok and E.~G.~D.~Cohen,
\emph{Deterministic lattice gas models},
Phys. Lett. A \textbf{133} (1988), nos.~7--8, 415--418.
\url{https://doi.org/10.1016/0375-9601(88)90927-9}.

\bibitem{Ryan}
K.~Ryan,
\emph{The Manhattan and Lorentz mirror models: a result on the
cylinder with low density of mirrors},
J. Stat. Phys. \textbf{185} (2021), no.~2, article 7.
\url{https://doi.org/10.1007/s10955-021-02837-8}.

\bibitem{RyanXXZ}
K.~Ryan,
\emph{The spin-\(1/2\) Heisenberg XXZ chain and the Lorentz mirror model
with loop weight \(2\)},
arXiv:2608.21306v1, 2026.
\url{https://arxiv.org/abs/2608.21306v1}.

\bibitem{ZiffKongCohen}
R.~M.~Ziff, X.~P.~Kong, and E.~G.~D.~Cohen,
\emph{Lorentz lattice-gas and kinetic-walk model},
Phys. Rev. A \textbf{44} (1991), no.~4, 2410--2428.
\url{https://doi.org/10.1103/PhysRevA.44.2410}.
\end{thebibliography}
\end{document}